\documentclass[10pt,conference]{IEEEtran}

\usepackage[dvips]{color}
\usepackage{epsf}

\usepackage{amsmath}
\usepackage{amsxtra}
\usepackage{graphicx}
\usepackage{epsfig}
\usepackage{latexsym}
\usepackage{amsfonts}
\usepackage{rawfonts}
\usepackage[latin1]{inputenc}
\usepackage[T1]{fontenc}
\usepackage{calc}
\usepackage{url}
\usepackage{enumerate}
\usepackage{color}
\usepackage{amssymb}
\usepackage{upref}
\usepackage{epic,eepic}
\usepackage{times}
\usepackage{cite}
\usepackage{dsfont}
\usepackage{mathrsfs}
\usepackage{verbatim}

\DeclareMathOperator*{\argh}{argmax}

\DeclareMathOperator*{\minim}{minimize}

\newtheorem{theorem}{\bf Theorem}

\newtheorem{definition}{\bf Definition}

\newcommand{\qed}{\nobreak \ifvmode \relax \else
  \ifdim\lastskip<1.5em \hskip-\lastskip
  \hskip1.5em plus0em minus0.5em \fi \nobreak
  \vrule height0.75em width0.5em depth0.25em\fi}

\newcounter{step}
\newlength{\totlinewidth}
  {\end{list}%
  \rule{\linewidth}{1pt}}
\newcounter{substep}

  {\end{list}}

\newlength{\aligntop}
\newlength{\alignbot}
 \makeatletter
\renewenvironment{align}{%
  \vspace{\aligntop}
  \start@align\@ne\st@rredfalse\m@ne
}{%
  \math@cr \black@\totwidth@
  \egroup
  \ifingather@
    \restorealignstate@
    \egroup
    \nonumber
    \ifnum0=`{\fi\iffalse}\fi
  \else
    $$%
  \fi
  \ignorespacesafterend%
  \vspace{\alignbot}\par\noindent
} \makeatother

\IEEEoverridecommandlockouts

\begin{document}
\title{\huge Matching Theory for Backhaul Management in Small Cell Networks with mmWave Capabilities}\vspace{-1em}

\author{
\authorblockN{Omid Semiari$^\dag$, Walid Saad$^\dag$, Zaher Dawy$^\ddag$, Mehdi Bennis$^\ast$}\vspace*{0.5em}
\authorblockA{\small $^\dag$Wireless@VT, Bradley Department of Electrical and Computer Engineering, Virginia Tech, Blacksburg, VA, USA, \\ Email: \protect\url{{osemiari, walids}@vt.edu}\\
\small $^\ddag$  Department of Electrical and Computer Engineering, American University of Beirut, Beirut, Lebanon, Email: \url{zd03@aub.edu.lb}\\
\small $^\ast$Centre for Wireless Communications-CWC, University of Oulu, Finland, Email: \url{bennis@ee.oulu.fi}
  }
%
\maketitle

\begin{abstract}
Designing cost-effective and scalable backhaul solutions is one of the main challenges for emerging wireless small cell networks (SCNs). In this regard, millimeter wave (mmW) communication technologies have recently emerged  as an attractive solution to realize the vision of a high-speed and reliable wireless small cell backhaul network (SCBN). In this paper, a novel approach is proposed for managing the spectral resources of a heterogeneous SCBN that can exploit simultaneously mmW and conventional frequency bands via carrier aggregation. In particular, a new SCBN model is proposed in which small cell base stations (SCBSs) equipped with broadband fiber backhaul allocate their frequency resources to SCBSs with wireless backhaul, by using aggregated bands. One unique feature of the studied model is that it jointly accounts for both wireless channel characteristics and economic factors during resource allocation. The problem is then formulated as a one-to-many matching game and a distributed algorithm is proposed to find a stable outcome of the game. The convergence of the algorithm is proven and the properties of the resulting matching are studied. Simulation results show that under the constraints of wireless backhauling, the proposed approach achieves substantial performance gains, reaching up to $30 \%$ compared to a conventional best-effort approach.
\end{abstract}

\section{Introduction}\label{intro}
The rapid proliferation of smartphones, tablets, and bandwidth-intensive wireless services has drastically strained the resources of existing cellular networks \cite{17}. To cope with this increasing traffic demand, small cell networks (SCNs) have been introduced as a promising approach to boost the coverage and capacity of current cellular networks \cite{Quek13}. These anticipated performance gains will be due to a cell size reduction and a dense deployment of small cell base stations (SCBSs). Nonetheless, this paradigm shift in cellular networks will lead to new challenges, specifically for the backhaul network design and management \cite{Quek13}. In fact, due to unplanned and dense deployment of small cells, it may not be economically viable for mobile network operators (MNOs) to deploy broadband fiber backhaul for every SCBS \cite{Hur13}. Therefore, new cost-effective backhaul solutions are required that provide sufficient bandwidth to support the backhaul traffic of the densely deployed SCNs.

One promising solution for backhauling in SCNs is via the use of \textit{millimeter wave (mmW) communications} \cite{Rangan14,Boccardi01,Hur13,Ghosh14}. The mmW spectral band ranging between 30-300 GHz will allow the deployment of highly directional antennas, as well as the possibility of transmission over a large available bandwidth, up to $1$ GHz. Such characteristics of mmW would help SCBSs to communicate over broadband backhaul links, while causing negligible interference to other, wireless backhaul links in the SCN. However, mmW signals attenuate rapidly with distance and cannot easily penetrate obstructing objects. Consequently, mmW frequencies are usually used for line of sight (LOS) point-to-point communication. Due to these limitations, mmW communication has to coexist with the conventional sub-6 GHz band which is attractive for non line of sight (NLOS) transmissions \cite{Rangan14}. Therefore, aggregating mmW band with sub-6 GHz band at the backhaul network provides an appealing and robust solution for small cell backhaul networks (SCBNs), especially, in dense urban areas. The concurrent transmissions at different frequency bands is achievable by using \emph{carrier aggregation }(CA) technique, introduced in LTE-Advanced standard \cite{Bhat12}. Another practical challenge for supporting an SCBN with mmW is the need for multi-MNO support \cite{Rangan14}. In fact, an SCBS with NLOS backhaul links may need the relaying support of neighboring SCBSs that belong to a second MNO. Therefore, multi-MNO operation will significantly increase the flexibility for mmW utilization.

In this regard, different wireless backhaul solutions have been proposed so far for SCNs. The authors in \cite{Liebl01} propose a fair resource allocation for the out-band relay backhaul links, enabled with CA. The proposed approach in \cite{Liebl01} aims to maximize the throughput fairness among backhaul and access links in LTE-Advanced relay system. In \cite{Yi01}, a backhaul resource allocation approach is proposed for LTE-Advanced in-band relaying. This approach optimizes resource partitioning between relays and macro users, taking into account both backhaul and access links quality. In \cite{Loumiotis01}, a dynamic backhaul resource allocation approach is proposed using evolutionary game theoretic approach. Instead of static backhaul resource allocation, the authors take into account the dynamics of users' traffic demand and allocate sufficient resources to the base stations, accordingly. Although interesting, the body of work in \cite{Liebl01,Boccardi01,Yi01,Loumiotis01} does not consider mmW communication at the SCBN and is primarily focused on modeling rather than resource allocation. In addition, it does not account for the effect of backhaul cost in modeling backhaul networks and solely focuses on physical layer metrics, such as links' capacity and fairness.

In wireless networks, the backhaul cost covers a substantial portion of the total cost of ownership (TCO) for MNOs \cite{Hur13}. In fact, it is economically inefficient for an individual MNO to afford the entire TCO of an independent backhaul network \cite{Hur13}. Therefore, MNOs may need to share their backhaul network resources with other MNOs that demand backhaul support for their SCBSs. Hence, beyond the technical challenges of backhaul management in SCNs, one must also account for the cost of sharing backhaul resources between MNOs.  The work in \cite{Sengupta01} proposes an economic framework to lease the frequency resources to different MNOs by using novel pricing mechanisms. In \cite{Mahloo01}, the authors propose a cost evaluation model for SCBNs. This work highlights the fact that integrating heterogenous backhaul technologies is mandatory to achieve a satisfactory performance in an SCBN. Moreover, they show that the TCO of an SCN is much higher than conventional cellular networks. Therefore, it is more critical to consider backhauling cost in SCBN design. The provided solutions in \cite{Sengupta01,Mahloo01} concentrate only on the economic aspects of the backhaul network, while a suitable SCBN model must integrate the cost constraints with the physical constraints of the wireless network, such as traffic demand.

The main contribution of this paper is to propose a novel resource management approach for CA-enabled small cell backhaul networks. In particular, we consider an SCBN in which MNOs seek to minimize their overall cost for using the backhaul, while maintaining the required data rates at the backhaul network. We propose a novel approach that allocates the backhaul frequency resource blocks to demanding SCBSs with no fiber backhaul support. We formulate the problem as a matching game and study its properties. To solve this game, we propose an algorithm that converges to a two-sided stable solution and guarantees the promised quality of service (QoS) at the SCBN. In addition, we compare our results with two other schemes: heuristic best-effort and random resource allocation. Our results show that by exploiting inexpensive and available mmW band, the proposed approach can achieve substantial gains compared to other resource management schemes. To our best knowledge, this paper is \emph{the first to develop a resource allocation solution that jointly accounts for the economics and wireless properties} of an SCBN that is endowed with both mmW and carrier aggregation capabilities.

The rest of this paper is organized as follows. Section ~\ref{model} describes the system model and formulates the problem. Section \ref{matching} presents our distributed approach and the proposed algorithm to solve the problem. Section \ref{simulation} provides the simulation results and Section \ref{conclusion} concludes the paper.

\section{System Model}\label{model}
Consider a small cell backhaul network composed of $K$ SCBSs from the set $\mathcal{K}$. Moreover, let $\mathcal{K}_1\subset \mathcal{K}$ denote the set of SCBSs that are connected to the core network via fiber links, which we refer to hereinafter as anchored BSs (A-BSs). The A-BSs are owned by a certain network operator $q_1$. The rest of the SCBSs, called demanding SCBSs (D-BSs) within the set $\mathcal{K}_2\subset \mathcal{K}$, are not equipped with broadband fiber backhaul. Instead, to serve their access link users at downlink, the D-BSs need the assistance of the A-BSs to receive the traffic from the core network. We assume that D-BSs belong to a second network operator $q_2$ and have to pay the A-BSs for receiving backhaul service. The system model is shown in Fig. \ref{fig1}. At the backhaul, each SCBS $k \in \mathcal{K} $ is enabled with CA technique and can operate in two different bands, simultaneously. The set of aggregated bands, known as component carriers (CCs), include millimeter wave (mmW) and sub-6 GHz bands, denoted by $\mathcal{N}_1$ and $\mathcal{N}_2$, respectively. The mmW band expands from $30$ GHz to $300$ GHz. Due to different propagation characteristics, the frequency bands in $\mathcal{N}_1$ are suitable for line of sight (LOS) point-to-point (P2P) links, while the sub-6 GHz band in $\mathcal{N}_2$ are more likely to be used for non-LOS (NLOS) communication. In addition, the backhaul signals are OFDM modulated so as to cope with the channel frequency selective fading and achieve high data rates. Here, $\mathcal{N}_1$ and $\mathcal{N}_2$ are composed, respectively, of $N_1$ and $N_2$ backhaul resource blocks (BRBs) and the total aggregated BRBs is $N=N_1+N_2$.\vspace{-.1cm}
\begin{figure}
  \centering
  \centerline{\includegraphics[width=8cm]{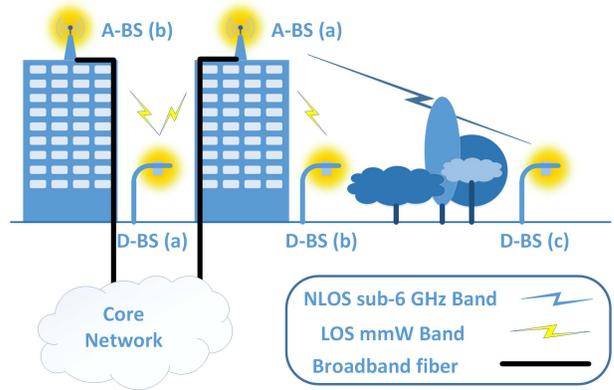}}
  \caption{A small cell backhaul network with carrier aggregation.}\vspace{-2em}
\label{fig1}
\end{figure}\vspace{0em}
\subsection{Propagation Model for mmW Band and Sub-6 GHz Band}\label{prop}
The main components of the propagation loss within the mmW frequency band include pathloss, atmospheric gaseous losses, and the rain attenuation loss.
For mmW propagation, we use the model proposed in \cite{Ghosh14} for the cellular urban scenarios:
\begin{align}\label{eq1aa}
L_{\text{mmW},\text{LOS}}=\beta+\alpha 10\log_{10}(d)+\chi.
\end{align}
In fact, (\ref{eq1aa}) is known to be the best linear fit to the propagation measurement in mmW frequency band\cite{Ghosh14}, where $\alpha$ is the slope of the fit and $\beta$, the intercept parameter, is the pathloss (in dB scale) for $1$ meter of distance. In addition, $\chi$ models the deviation in fitting (in dB scale) which is a Gaussian random variable with zero mean and variance $\xi^2$. Extensive field measurement have shown that large bandwidth mmW networks in urban settings are tend to be noise limited \cite{Rangan14,Ghosh14}. This is due to the fact that mmW band allows the implementation of dense antenna based network architecture with high directional gains.

However, this is not true for sub-6 GHz band, as the antenna beams are much wider. Hence, the propagation in this band is interference limited. Assuming no power control, the signal-to-interference-plus-noise ratio (SINR) at the receiving SCBS in the sub-6 GHz band is given by:
\begin{align}\label{eq6}
\gamma_{k_1nk_2}=\frac{\psi_{k_1}|h_{k_1nk_2}|^2}{\sum_{j \in \mathcal{K}_1\backslash k_1} \psi_{j}|h_{jnk_2}|^2+\sigma^2},
\end{align}
where $h_{k_1nk_2}$ represents the channel state of the link between A-BSs $k_1$ and D-BSs $k_2$ over BRB $n \in \mathcal{N}_2$. In addition, $\psi_{k}$ and $\sigma^2$ denote, respectively, the transmit power of A-BS $k$ and the variance of the receiver's Gaussian noise.

\subsection{Problem Formulation}\label{problem}
We consider the SCBN during downlink transmissions at the radio side. In the downlink, the A-BSs forward the traffic from the core network to the D-BSs by using the available CCs. For example, in Fig. \ref{fig1}, D-BS (b) can have a backhaul link, either through LOS mmW with A-BS (a) or NLOS sub-6 GHz link with A-BS (b). However, D-BS (c) may only have NLOS backhaul link with A-BS (a). Since the A-BSs and D-BSs belong to two different MNOs, network operator $q_2$ has to make a payment to network operator $q_1$ for allowing its A-BSs to relay the backhaul traffic of $q_2$'s network. For tractability, we assume that $q_2$ pays $q_1$ for each BRB that A-BSs allocate to D-BSs for relaying the backhaul traffic.

Compared to the sub-6 GHz band, mmW links are more vulnerable to the signal blockage and are less congested. Therefore, MNO $q_1$ must intuitively provide cheaper service over mmW band than sub-6 GHz band. Let $\mathcal{P}$ be a set of pairs $p_{k_1}=(p_{k_1,1},p_{k_1,2}), \forall k_1 \in \mathcal{K}_1$ , where $p_{k_1,i}, i \in \{1,2\}$ denotes the price that $q_1$ will charge $q_2$ to allocate a BRB of A-BS $k_1$ in the $\mathcal{N}_i$ band to a D-BS. In this model, the goal of MNO $q_2$ is to minimize the cost of its SCBN by exploiting CA capabilities, while ensuring the required backhaul QoS.
To simplify the notation, we use $\mathcal{N}=\mathcal{N}_1 \cup \mathcal{N}_2$, to represent the set of all BRBs in both CCs.

Given the BRB price set by MNO $q_1$, we need to find the optimal backhaul resource allocation, such that every D-BS $k_2$ can meet its required data rate $R_{k_2}^d$. We assume that every D-BS $k_2$ is associated with a maximum budget $B_{k_2}$ that allows $k_2$ to use BRBs from A-BSs. Therefore, MNO $q_2$ has to solve the following optimization problem:
\begin{IEEEeqnarray}{rCl}\label{eq1}
&&\minim_{\eta^*} \, C_{\text{total}}=\!\!\!\sum_{i \in \{0,1\}}\sum_{n \in \mathcal{N}_i}\sum_{k_1 \in \mathcal{K}_1}\sum_{k_2\in \mathcal{K}_2 } p_{k_1,i}\eta_{k_1nk_2}\!\!\IEEEyessubnumber\label{1a}\\
&&\textrm{s.t.},\,\, R_{k_2}(\eta,\gamma_{k_1nk_2})\geq R_{k_2}^d,\,\,\,\,\, \forall k_2 \in \mathcal{K}_2,\IEEEyessubnumber\label{1b}\\
&&\sum_{i \in \{0,1\}}\sum_{n \in \mathcal{N}_i}\sum_{k_1 \in \mathcal{K}_1}p_{k_1,i}\eta_{k_1nk_2}\leqslant B_{k_2},\, \forall k_2 \in \mathcal{K}_2,\IEEEyessubnumber\label{1c}\\
&&\sum_{n \in \mathcal{N}}\sum_{k_2\in \mathcal{K}_2 }\eta_{k_1nk_2}\leq N, \forall k_1 \in \mathcal{K}_1,\IEEEyessubnumber\label{1d}\\
&&\sum_{k_2 \in \mathcal{K}_2}\eta_{k_1nk_2}\leq 1, \forall n \in \mathcal{N}, \forall k_1 \in \mathcal{K}_1,\IEEEyessubnumber\label{1e}\\
&&\eta_{k_1nk_2} \in \{0,1\},\IEEEyessubnumber\label{1f}
\end{IEEEeqnarray}
where the backhaul resource allocation variable is such that $\eta_{k_1nk_2}=1$, if A-BS $k_1$ allocates BRB $n \in \mathcal{N}$ to D-BS $k_2$, and $\eta_{k_1nk_2}=0$, otherwise. The total achievable backhaul rate for SCBS $k_2$ can be written as:
\begin{align}\label{eq:new}
R_{k_2}(\eta,\gamma_{k_1nk_2})=\!\!\!\!\sum_{k_1\in \mathcal{K}_1}\sum_{n\in \mathcal{N}}\omega_n\log_2\left(1+\gamma_{k_1nk_2}\right)\eta_{k_1nk_2},
\end{align}
Here, $\gamma_{k_1nk_2}$ determines signal-to-noise ratio, if $n \in \mathcal{N}_1$, otherwise, $\gamma_{k_1nk_2}$ is the SINR obtained from (\ref{eq6}). Moreover, $\omega_n$ denotes the bandwidth of BRB $n$ which is different for each CC. The SINR/SNR distribution is a function of cells' location, transmit power, backhaul resource allocation and the propagation model of the carrier.


With this in mind, we can translate the optimization problem in (\ref{1a})-(\ref{1f}) into the assignment of D-BSs $k_2 \in \mathcal{K}_2$ to the {\cal N} BRBs of A-BSs $k_1 \in \mathcal{K}_1$, in both carriers. During this assignment, the network has to minimize the cost of backhaul service for D-BSs, while taking into account the rate demands and budget constraints. This problem is a 0-1 integer programming, which is a satisfiability problem as such one of Karp's 21 NP-complete problems. Therefore, it is difficult to solve (\ref{1a})-(\ref{1f}) via classical optimization approaches \cite{Karp72}.

To address this problem, we propose a distributed approach that can provide a suitable solution for the proposed backhaul resource allocation problem. Given the BRB price set by MNO $q_1$, our goal is to develop a \textit{self-organizing, decentralized} approach, that can guarantee the QoS of the D-BSs, while taking into account their limited budget.

\section{Matching Theory for Backhaul Management}\label{matching}
To solve the proposed SCBN resource management problem, we propose a novel approach based on matching theory. Matching theory is a promising approach to realize distributed resource management in wireless networks \cite{Roth92,eduard11}. The merit of matching theory is that it allows to provide a decentralized solution with tractable complexity for combinatorial allocation problems, as such that of (\ref{1a})-(\ref{1f}).

In essence, a matching game is a two-sided assignment problem between two disjoint sets of players in which the players of each set are interested to be matched to the players of the other set, according to some \textit{preference relations}. A preference relation $\succ$ is defined as a complete, reflexive, and transitive binary relation between the elements of a given set. Here, we denote $\succ_m$ as the preference relation of player $m$ and denote $a\succ_m b$, if player $m$ prefers $a$ more than $b$.

In the proposed backhaul resource management problem, the preference relations would depend on both the rate and the cost of the backhaul link. Matching theory allows to specify preference relations, by defining individual utility functions for A-BSs and D-BSs, which simplifies accounting for the cost constraints within the resource allocation problem in (\ref{1a})-(\ref{1f}).

\subsection{Backhaul Management as a Matching Game}
For the studied backhaul management problem, the matching game will occur between the D-BSs in $\mathcal{K}_2$ and the BRBs in the set $\mathcal{N}$. The BRBs are controlled by their corresponding A-BS throughout the matching process. We introduce a two-sided \emph{one-to-many matching game} in which each BRB will be matched to maximum one D-BS, while a D-BS can be assigned to one or more BRBs. Formally, we can define the matching game as follows:
\begin{definition}\label{def1}
Given the two disjoint finite sets of players $\mathcal{K}_2$ and $\mathcal{N}$, a \textit{matching game} $\mu$ is defined as a function from $\mathcal{K}_2  \rightarrow \mathcal{K}_1 \times \mathcal{N}$ from which we have
\begin{itemize}
\item[1.] $\forall n \in \mathcal{N}, \mu(n) \in \mathcal{K}_2$,
\item[2.] $\forall k_2 \in \mathcal{K}_2, \mu(k_2) \subseteq \mathcal{K}_1 \times\mathcal{N}$,
\item[3.] $\mu(n)=k_2$, if and only if  $n \in \mu(k_2)$.
\end{itemize}
\end{definition}
In matching theory, the \emph{quota} of a player is defined as a maximum number of players, that a player can be matched to. Here, the quota of BRBs is set to one, while there is no pre-determined quota for D-BSs. In addition, we have $(k_2,n) \in \mu$, if BRB $n$ is assigned to D-BS $k_2$ through matching $\mu$ and $(k_2,n) \notin \mu$, otherwise.

In this game, given the BRB prices, each D-BS in $\mathcal{K}_2$ aims to reduce the cost of backhauling, while meeting its required data rate demand. In other words, for each D-BS $k_2 \in \mathcal{K}_2$,
\begin{IEEEeqnarray}{rCl}\label{5}
&\minim_{\eta} \, C_{k_2}=\sum_{i \in \{0,1\}}\sum_{n \in \mathcal{N}_i}\sum_{k_1 \in \mathcal{K}_1}p_{k_1,i}\eta_{k_1nk_2}, \IEEEyessubnumber\label{5a}\\
&\textrm{s.t.},\,R_{k_2}(\eta,\gamma)\geq R_{k_2}^d,\,\IEEEyessubnumber\label{5b}\\
&C_{k_2} \leqslant B_{k_2}.\,\IEEEyessubnumber\label{5c}
\end{IEEEeqnarray}
In fact, (\ref{5a})-(\ref{5c}) is the distributed version of the main optimization problem in (\ref{eq1}). The one-to-many matching game will automatically satisfy the allocation constraints in (\ref{1d})-(\ref{1f}).

The problem in (\ref{5}) implies that the preference relations of D-BS $k_2$, $\succ_{k_2}$, depends on both the channel state of each backhaul links and the price of BRBs. Therefore, we must introduce suitable utility functions for each player set, using which the players can define their preference relations. Let $V_{k_2,n}(.)$ and $U_{n,k_2}(.)$ denote, respectively, the utility functions of players in $\mathcal{K}_2$ and $\mathcal{N}$. Hence, D-BS $k_2$ prefers BRB $n_1$ to $n_2$, $n_1\succ_{k_2} n_2$, if and only if $V_{k_2,n}(n_1)>V_{k_2,n}(n_2)$.

The SCBS $k_2 \in \mathcal{K}_2$ will yield a utility $V_{k_2,n}$ to BRB $n \in \mathcal{N}$ of A-BS $k_1 \in \mathcal{K}_1$ based on the following utility function:
\begin{align}\label{eq2}
&V_{k_2,n}(\gamma_{k_2nk_1},p_{k_1})= \omega_n\log(1+\gamma_{k_2nk_1})-\zeta p_{k_1,i},
\end{align}
where $i=1$, if $n$ is a BRB of the mmW band, otherwise $i=2$. In (\ref{eq2}), $\zeta$ is a weighting parameter that controls the tradeoff between the cost and the link's achievable rate. By using this utility, the D-BSs will seek to be matched to the cheapest BRBs through the links with higher achievable rates.

Furthermore, the A-BSs aim to assign each of their BRBs to a particular D-BS in $\mathcal{K}_2$, from which the link's achievable rate is maximized. Therefore, the utility that BRB $n$ belonging to A-BS $k_1$ assigns to SCBS $k_2$, $U_{n,k_2}$, is given by
\begin{align}\label{eq3}
&U_{n,k_2}(\gamma_{k_2nk_1},R_{D,k_2})= \omega_n\log(1+\gamma_{k_2nk_1}).
\end{align}

Having formulated the backhaul management problem as a matching game, our next step is to provide a solution that allows the D-BSs to take advantage of CA capabilities and satisfy their rate demands.
\subsection{Proposed Backhaul Management Algorithm and Solution}
To solve the proposed matching game, one suitable concept is the so-called two-sided \textit{stable matching} between the D-BSs and BRBs, defined as follows \cite{Roth92}:
\begin{definition}
A pair $(k_2,n) \notin \mu$ is said to be a \textit{blocking pair} of the matching $\mu$, if and only if $n \succ_{k_2} \mu(k_2)$ and $k_2 \succ_n \mu(n)$.
Matching $\mu^*$ is stable, if there is no blocking pair.
\end{definition}

Under a stable matching $\mu^*$, one can ensure that A-BSs will not reallocate their BRBs to other D-BSs. Hence, it ensures the promised QoS to the D-BSs, or more specifically, to the MNO $q_2$.

The classical approach used to solve a matching game is by adopting the so-called \textit{deferred acceptance} algorithm which is known to converge to a stable matching, given fixed preference relations and fixed quotas \cite{Roth92,eduard11}. In our problem, however, we cannot fix the quotas for D-BSs, since it is neither spectral nor cost efficient to allocate more BRBs than a D-BS requires to manage its backhaul traffic.

Therefore, the approach of \cite{Roth92} and \cite{eduard11} are inapplicable and a new algorithm must be developed to find a stable matching. To this end, we propose the algorithm in Table \ref{tab:algo}. Prior to the start of the assignment algorithm, the D-BSs obtain the channel state from the control signals sent by the A-BSs, along with the price values $p_k$. Then, each D-BS $k_2$ initializes the set $\mathcal{S}_{k_2}=\mathcal{K}_1 \times \mathcal{N}$ as the list of all BRBs that D-BS $k_2$ may apply to them.
\begin{table}[!t]
  \centering
  \caption{
    \vspace*{-0em}Proposed Backhaul Management Algorithm}\vspace*{-0.9em}
    \begin{tabular}{p{\columnwidth}}
      \hline \vspace*{-0em}
      \textbf{Inputs:}\,\,$\mathcal{K}_1,\mathcal{K}_2, \mathcal{N}_1, \mathcal{N}_2, \mathcal{P}$\\
\hspace*{1em}\textit{Initialize:}   \vspace*{0em}
D-BSs in $\mathcal{K}_2$ find the BRB prices and determine the channel states, throughout receiving control signal from A-BSs in $\mathcal{K}_1$. Initialize the set $\mathcal{S}_{k_2}$ to the total available BRBs of the A-BSs.

\hspace*{0em} \textit{Stage 1:}\begin{itemize}\vspace*{-0em}
\item[] \hspace*{0em}(a) D-BSs rank all the BRBs in $\mathcal{K}_1 \times \mathcal{N}$, using (\ref{eq2}).
\item[] \hspace*{0em}(b) Every D-BS $k_2 \in \mathcal{K}_2$ applies for its most preferred BRB $n^*$ from $\mathcal{S}_{k_2}$ and removes $n^*$ from $\mathcal{S}_{k_2}$.
\end{itemize}\vspace*{-0cm}

\hspace*{0em} \textit{Stage 2:}\begin{itemize}\vspace*{-0em}
\item[] \hspace*{0em}(a) BRBs $n \in \mathcal{K}_1 \times \mathcal{N}$ calculate utility of each applicant using (\ref{eq3}).
\item[] \hspace*{0em}(b) Among all applicants and its current match $\mu(n)$, BRB $n$ accepts the most preferred D-BS and rejects the rest.
\item[] \hspace*{0em}(c) If $k_2$ is accepted by BRB $n$ belonging to $i$-th CC of $k_1$, set $R_{k_2}=R_{k_2}+w\log_2(1+\gamma_{k_1nk_2})$ and $C_{k_2}=C_{k_2}+p_{k_1,i}$. In addition, if $k_2$ is rejected by a BRB $n \in \mu(k_2)$, set $R_{k_2}=R_{k_2}-w\log_2(1+\gamma_{k_1nk_2})$ and $C_{k_2}=C_{k_2}-p_{k_1,i}$.
\end{itemize}\vspace*{-0cm}
\hspace*{0em}\textit{\textbf{repeat} Stage 1 and Stage 2 \textbf{until} for every $k_2 \in \mathcal{K}_2$, either $R_{k_2}>R_{k_2}^d$ , or $\mathcal{S}_{k_2}=\emptyset$, or $C_{k_2}>B_{k_2}$.}\vspace*{0em}\\

\hspace*{0em}\textit{Stage 3:}\begin{itemize}\vspace*{-0em}
\item[]   \hspace*{0em}The actual data traffic will be forwarded from A-BSs to the SCBSs $k_2 \in \mathcal{K}_2$ over BRBs $n \in \mu(k_2)$.
\end{itemize}
\hspace*{0em}\textbf{Output:}\,\,Stable matching $\mu^*$\vspace*{0em}\\
   \hline
    \end{tabular}\label{tab:algo}\vspace{-0.5cm}
\end{table}

In Stage 1, the D-BSs rank the BRBs, based on the given prices and the achievable rate of the links, as per (\ref{eq2}). Furthermore, each D-BS $k_2$ applies for its most preferred BRB, $n^*=\argh_{n}V_{k_2,n}(\gamma_{k_2nk_1}, p_{k_1})$, from the set $\mathcal{S}_{k}$, then removes $n^*$ from $\mathcal{S}_{k}$, in order to avoid applying for the same BRB multiple times.

In Stage 2, BRBs receive the applications and calculate the corresponding utility values, using (\ref{eq3}). Each BRB accepts the most preferred D-BS, $k_2^*=\argh_{k_2} U_{n,k_2}(\gamma_{k_2nk_1})$, and rejects previous $\mu(n)$ as well as the rest of the applicants. If accepted by BRB $n$, the SCBS $k$ adds the achievable rate from that link to $R_{k_2}$. In addition, it subtract the cost of BRB $n$ from its total cost $C_{k_2}$. In case D-BS $k_2$ has been accepted by BRB $n$ and now is rejected, $k_2$ will subtract the achievable rate of this link from $R_{k_2}$ and adds the price of $n$ to its budget.

The players iterate based on Stage $1$ and Stage $2$, until we reach a stage in which, for every D-BS, either $R_{k_2}>R_{k_2}^d$, or $\mathcal{S}_{k_2}=\emptyset$, or $C_{k_2}>B_{k_2}$. Following Stage 2, in Stage 3, the actual data traffic will be forwarded from A-BSs to the SCBSs $k_2 \in \mathcal{K}_2$ over BRBs $n \in \mu(k_2)$.

\begin{theorem}
The proposed algorithm in Table \ref{tab:algo} is guaranteed to converge to a two-sided stable matching between D-BSs and BRBs.
\end{theorem}
\begin{proof}
The convergence of the proposed algorithm in Table \ref{tab:algo} is guaranteed, since a D-BS never applies for a certain BRB twice. Hence, in the worst case, all D-BSs will apply for all BRBs once, which results $\mathcal{S}_{k_2}=\emptyset, \forall k_2 \in \mathcal{K}_2$. Next, we show that once the algorithm converges, the resulting backhaul resource allocation between D-BSs and BRBs is two-sided stable. Lets assume that there exists a pair $(k_2,n) \notin \mu$ that blocks the outcome of the algorithm. Since the algorithm has converged, we can conclude that at least one of the following cases is true about $k_2$: $R_{k_2}>R_{k_2}^d$ , or $\mathcal{S}_{k_2}=\emptyset$, or $C_{k_2}>B_{k_2}$.

The first case is $R_{k_2}>R_{k_2}^d$. This means $k_2$ does not need to add more BRBs to $\mu(k_2)$. In addition, $k_2$ would not replace any of $n' \in \mu(k_2)$ with $n$, since $n' \succ_{k_2} n$. Otherwise, $k_2$ would apply earlier for $n$. If $k_2$ has applied for $n$ and got rejected, this means $\mu(n) \succ_n k_2$, which contradicts $(k_2,n)$ to be a blocking pair. Analogous to the first case, $\mathcal{S}_{k_2}=\emptyset$ implies that $k_2$ has got rejected by $n$, which means $\mu(n) \succ_n k_2$ and $(k_2,n)$ cannot be a blocking pair.

Finally, $C_{k_2}>B_{k_2}$ means $k_2$ cannot afford more BRBs. Moreover, $k_2$ would not replace any of $n' \in \mu(k_2)$ with $n$, due to the same reasoning as case 1. This proves the theorem.
\end{proof}\vspace{-.2cm}
\section{Simulation Results and Discussion}\label{simulation}
\begin{table}[!t]
  \centering
  \caption{
    \vspace*{-0em}Simulation parameters}\vspace*{-0em}
\begin{tabular}{|c|c|c|}
\hline
\bf{Notation} & \bf{Parameter} & \bf{Value} \\
\hline
$K$ & Total number of small cells & $10$\\
\hline
$K_1$ & Total number of A-BSs & $2$\\
\hline
$(N_1,N_2)$ & Number of BRBs& $(192,100)$ \\
\hline
$(\omega_1,\omega_2)$ & Bandwidth of BRBs& $(4.86 $MHz$, 480 $KHz$)$ \\
\hline
$\psi$ & Transmit power& $1$ W \\
\hline
$\xi$ & Standard deviation of mmW path loss& $4.1$ \cite{Ghosh14} \\
\hline
$\alpha$ & Path loss exponent& $2$ \cite{Ghosh14}\\
\hline
$\beta$ & Path loss at $1$ m& $70$ dB \\
\hline
$R_k^d$ & Rate demand& $100$ Mbits/s \\
\hline
$p_{k_1}$ & Price of BRBs& $(0.1,10)$ \\
\hline
$\sigma^2$ & Noise power& $-90$ dBm \\
\hline
\end{tabular}\label{tab2}
\end{table}
For our simulations, we consider an SCBN with two A-BSs and eight D-BSs. All SCBSs are distributed randomly within a $2 \,\text{km} \times 2 \, \text{km}$ square area. The center frequency of mmW band and sub-6 GHz band is $73$ GHz and $5.8$ GHz, respectively. The simulation parameters are summarized in Table \ref{tab2}. The wireless mmW links are noise limited, while the channels in sub-6 GHz band experience Rayleigh fading with the path loss exponent set to three. The simulations are averaged over a large number of iterations for different channel fading and SCBS locations. We compare the performance of the proposed approach with a heuristic best-effort algorithm as well as the random allocation, all of which are enabled with CA. The best-effort algorithm allocates BRBs to D-BSs based on the maximum achievable rate, while the random BRB allocation simply allocates BRBs to D-BSs, randomly.

\begin{figure}
  \centering
  \centerline{\includegraphics[width=\columnwidth]{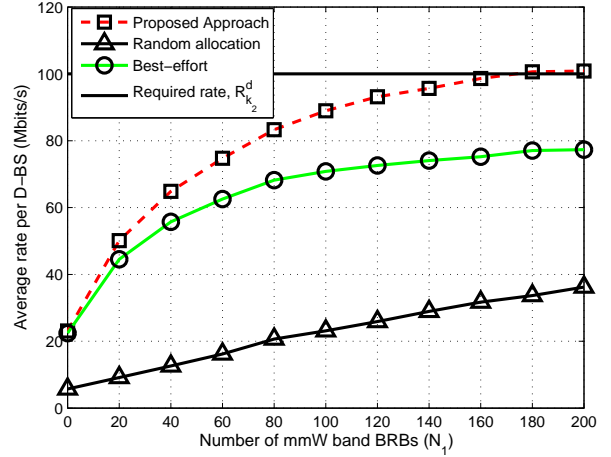}}\vspace{-.1cm}
  \caption{Average rate per D-BS vs the number of mmW BRBs $(N_1)$. The $p_{k_1}=(0.1, 10)$ for both A-BSs, $\zeta=1$ Mbits/s/price unit, and $N_2=100$. The rate demand is $R_D=100 \text{Mbits/s}$ for each D-BS. The budget of each D-BS is limited to $60$.}\vspace{-1em}
\label{2final}
\end{figure}

\begin{figure}
  \centering
  \centerline{\includegraphics[width=\columnwidth]{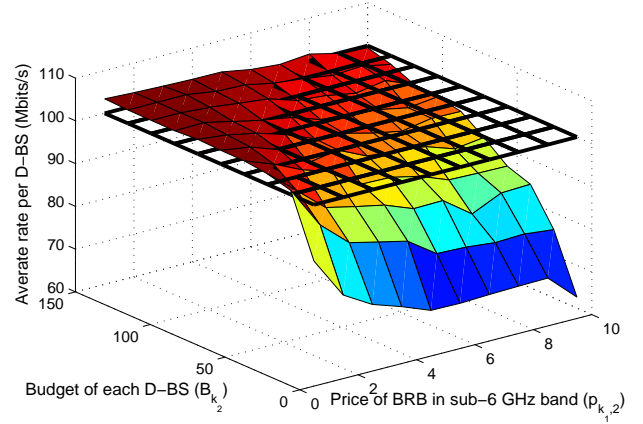}}\vspace{-.2cm}
  \caption{Average rate per D-BS vs different budgets $B_{k_2}$ and BRB prices of sub-6 GHz band, $p_{k,2}$. The weighting parameter $\zeta=0.1 $ Mbits/s/price unit. The constant mesh grid shows the required rate $R_{k}^d=100$ Mbits/s. }\vspace{-2em}
\label{fig3}
\end{figure}

In Fig. \ref{2final}, we show the average rate per D-BS resulting from the proposed algorithm as the number of mmW band BRBs varies. The results are compared with two other algorithms, the best-effort scheme and the random allocation scheme. Here, constant solid line at $100$ Mbits/s indicates the rate requirement of each D-BS and the budget per D-BS is limited to $60$. The achieved rates for all three approaches increase with the increase in mmW bandwidth, since more BRBs will be available for each D-BS. Fig. \ref{2final} shows that the proposed approach achieves substantial performance gains, reaching up to $30 \%$ and $190 \%$, respectively, compared to best-effort approach and random allocation for the network with $N_1=180$. Due to the limited budget, the proposed approach will shift from using the sub 6-GHz band towards exploiting inexpensive and available mmW links, as the number of mmW BRBs increases. In contrast, the greedy best-effort algorithm disregards the limited budget of D-BSs and allocates BRBs only based on the quality of the links. Overall, Fig. \ref{2final} reveals the potentials of aggregating mmW band with existing backhaul technologies, specially when the wireless backhaul cost is a limiting constraint for MNOs.

Fig. \ref{fig3} shows the average rate per D-BS resulting from the proposed algorithm, versus different budget values and BRB prices. The constant mesh grid shows the required rate $R_{k}^d=100$ Mbits/s. In Fig. \ref{fig3}, we can clearly see that for low price values, the required rate is achieved, almost for all budget values. In addition, comparing Fig. \ref{fig3} with Fig. \ref{2final}, we can see that choosing a smaller value for $\zeta$ in Fig. \ref{fig3} has decreased the average rate. That is due to the fact that $\zeta$ controls the effect of BRB prices in the proposed backhaul resource allocation. Therefore, for more limited budgets, D-BSs should pick larger values for $\zeta$, to be able to meet the required rate.

Fig. \ref{fig4} shows the convergence time of the proposed algorithm, as the number of SCBSs grows. This figure compares the number of iterations for two different values of rate demands $R_{k}^d=100$ Mbits/s, and $R_{k}^d=50$ Mbits/s. The number of iterations of the proposed algorithm is expected to increase with the network size, as clearly seen in Fig. ~\ref{fig4}. However, this figure shows that the convergence of the proposed distributed algorithm does not exceed $400$ iterations for a very dense SCBN with $20$ SCBSs. Indeed, the convergence time increases linearly as the network size grows, thus, yielding a reasonable scalability and complexity. In addition, the convergence time for $R_{k}^d=50$ Mbits/s is almost half of the this value when $R_{k}^d=100$ Mbits/s. Thus, Fig. \ref{fig4} shows that the traffic at the radio access network substantially affects the convergence time of the backhaul resource allocation.
\begin{figure}
  \centering
  \centerline{\includegraphics[width=\columnwidth]{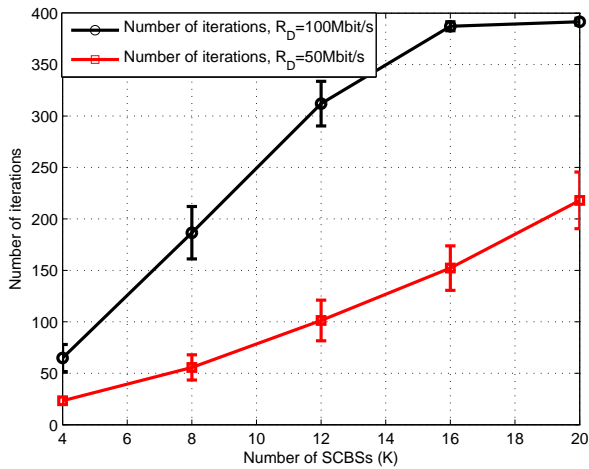}}\vspace{-.25cm}
  \caption{The number of iterations with $95 \%$ confidence vs $K$ with $K_1=2$, and for two different rate demands, $R_k^d=100$ Mbits/s and $R_k^d=50$ Mbits/s. The $p_{k_1}=(0.1,10)$ for both A-BSs. The budget of D-BSs is limited to $60$.}\vspace{-2em}
\label{fig4}
\end{figure}
\section{Conclusion}\label{conclusion}
In this paper, we have proposed a novel distributed backhaul management approach for small cell networks enabled with backhaul carrier aggregation. We have formulated the problem as a one-to-many matching game which provides a mean to solve the proposed backhaul management problem with manageable complexity. To solve the game, we have proposed an algorithm to match D-BSs to the A-BS controlled backhaul resource blocks and proved the two-sided stability of its outcome. Simulation results have shown that the proposed algorithm outperforms the heuristic best-effort and random allocation approaches and achieves higher rates, subject to the budget constraints of the small cells.

The results provide insightful guidelines on designing wireless backhaul networks, in which the network operators may need to combine the monetary aspects with the physical constraints of the network. In addition, the results emphasize on the potential gains of utilizing available mmW band in the backhaul network, along with the congested sub-6 GHz band.

\bibliographystyle{ieeetr}
\bibliography{references}
\end{document}